\documentclass[11pt]{scrartcl}

\usepackage{url}                      
\usepackage{graphicx,color} 
\usepackage{mathrsfs}
\usepackage{latexsym}
\usepackage{amssymb, amstext, amsgen}
\usepackage[centertags,intlimits]{amsmath}
\usepackage{amsthm}
\usepackage{a4wide}
\usepackage[latin1]{inputenc}         
\usepackage[ruled,vlined]{algorithm2e}

\def\shu{\mathbin{\mathchoice
{\rule{.3pt}{1ex}\rule{.3em}{.3pt}\rule{.3pt}{1ex}
\rule{.3em}{.3pt}\rule{.3pt}{1ex}}%
{\rule{.3pt}{1ex}\rule{.3em}{.3pt}\rule{.3pt}{1ex}
\rule{.3em}{.3pt}\rule{.3pt}{1ex}}%
{\rule{.2pt}{.7ex}\rule{.2em}{.2pt}\rule{.2pt}{.7ex}
\rule{.2em}{.2pt}\rule{.2pt}{.7ex}}%
{\rule{.3pt}{1ex}\rule{.3em}{.3pt}\rule{.3pt}{1ex}
\rule{.3em}{.3pt}\rule{.3pt}{1ex}}\mkern2mu%
}}

\newtheorem{prop}{Proposition}
\newtheorem{Remark}{Remark}
\newtheorem{Lemma}{Lemma}
\newtheorem{coro}{Corollary}

\begin{document}

\title{On Packing Colorings of Distance Graphs}

\author {Olivier Togni\\\\
\textit{LE2I, UMR CNRS 5158}\\ \textit{ Universit\'e de Bourgogne, 21078 Dijon cedex, France}\\ \texttt{\small\
Olivier.Togni@u-bourgogne.fr}}
\maketitle

\begin{abstract}
The {\em packing chromatic number} $\chi_{\rho}(G)$ of a graph $G$ is the least integer $k$ for which there exists a mapping $f$ 
from $V(G)$ to $\{1,2,\ldots ,k\}$ such that any two vertices
of color $i$ are at a distance of at least $i+1$. This paper studies the packing chromatic number of infinite distance graphs $G(\mathbb{Z},D)$, 
i.e. graphs with the set $\mathbb{Z}$ of integers as vertex set, with two distinct vertices $i,j\in \mathbb{Z}$ being adjacent if and only if $|i-j|\in D$. 
We present lower and upper bounds for $\chi_{\rho}(G(\mathbb{Z},D))$, showing that for finite $D$, the packing chromatic number is finite. 
Our main result concerns distance graphs with $D=\{1,t\}$ for which we prove some upper bounds on their packing
chromatic numbers, the smaller ones being for $t\geq 447$: $\chi_{\rho}(G(\mathbb{Z},\{1,t\}))\leq 40$ if $t$ is odd and 
$\chi_{\rho}(G(\mathbb{Z},\{1,t\}))\leq 81$ if $t$ is even.

\paragraph{Keywords:} graph coloring; packing chromatic number; distance graph.
\end{abstract}



\section{Introduction }
Let $G$ be a connected graph and let $k$ be an integer, $k\geq 1$.
A {\em packing $k$-coloring} (or simply a packing coloring) of a graph $G$ is a mapping $f$ from $V(G)$ to $\{1,2,\cdots ,k\}$ such that for any two distinct vertices $u$ and $v$, if $f(u)=f(v)=i$, then $\text{dist}(u,v)> i$, where $\text{dist}(u,v)$ is the distance between $u$ and $v$ in $G$ (thus vertices of color $i$ form an $i$-packing of $G$). The \textit{packing chromatic number} $\chi_{\rho}(G)$ of
$G$ is the smallest integer $k$ for which $G$ has a packing $k$-coloring.

This parameter was introduced recently by Goddard et al.~\cite{GoBro} under the name of {\em broadcast chromatic number} and the authors showed that deciding 
whether $\chi_{\rho}(G)\leq 4$ is NP-hard. Fiala and Golovach~\cite{FiCo} showed that the packing coloring problem is NP-complete
for trees.
Bre{\v{s}}ar et al.~\cite{BrePa} studied the problem on Cartesian products graphs, hexagonal lattice and trees, using the name of packing chromatic number.
Other studies on this parameter mainly concern infinite graphs, with a natural question to be answered : 
{\em does a given infinite graph have finite packing chromatic number ?}
Goddard et al. answered this question affirmatively for the infinite two dimensional square grid by showing $9\leq \chi_{\rho}\leq 23$. 
The lower bound was later improved to $10$ by Fiala et al.~\cite{FiPa} and then to $12$ by Ekstein et al.~\cite{EkPa}.
The upper bound was recently improved by Holub and Soukal~\cite{SoPa} to $17$.
Fiala et al.~\cite{FiPa} showed that the infinite hexagonal grid has packing chromatic number 7; 
while both the infinite triangular lattice  and the 3-dimensional square lattice were shown to admit
no finite packing coloring by Finbow and Rall~\cite{FiLa}.
Infinite product graphs were considered by Fiala et al.~\cite{FiPa} who showed that the product of a finite path (of order at least two) 
with the 2-dimensional square grid has infinite packing chromatic number while the product of the infinite path and any finite graph has 
finite packing chromatic number.

The (infinite) \textit{distance graph} $G(\mathbb Z,D)$ with
distance set $D=\{d_1,d_2,\ldots ,d_k\}$, where $d_i$ are positive integers, has the set $\mathbb{Z}$ of integers as vertex set, with two distinct vertices $i,j\in \mathbb{Z}$ being adjacent if and only if $|i-j|\in D$. 
The {\it finite distance graph} $G_{n}(D)$ is the subgraph of $G(\mathbb Z,D)$ induced by vertices $0,1,\ldots ,n-1$. 
To simplify, $G(\mathbb Z, \{d_1,d_2, \ldots ,d_k\})$ will also be denoted as $D(d_1,d_2, \ldots ,d_k)$ and $G_n(\{d_1,d_2, \ldots ,d_k\})$ as $D_n(d_1,d_2, \ldots ,d_k)$.

The study of distance graphs was initiated by Eggleton et al.~\cite{Egg}.
A large amount of work has focused on colorings of distance graphs~\cite{EggDis,VoDis,BarDis,Liu,LiuZhu,SteOdd}, but other parameters have also been studied on distance graphs, like the feedback vertex set problem~\cite{KT}.

The aim of this paper is to study the packing chromatic number of infinite distance graphs, with particular emphasis on the case $D=\{1,t\}$. 
In Section~\ref{s2}, we bound the packing chromatic number of the infinite path power (i.e. infinite distance graph with $D=\{1,2,\ldots, t\}$).
Section~\ref{s3} concerns packing colorings of distance graphs with $D=\{1,t\}$, 
for which we prove some lower and upper bounds on the number of colors (see Proposition~\ref{p0}).
Exact or sharp results for the  packing chromatic number of some other 4-regular distance graphs are presented in Section~\ref{ssmall}. 
Section~\ref{s4} concludes the paper with some remarks and open questions.

Our results about the packing chromatic number of $G(\mathbb Z,D)$ for some small values of $D$ (from
Sections~\ref{s2} and \ref{ssmall}) are summarized in Table~\ref{tb1}. 
\begin{table}[ht]
 \centering
$$\begin{array}{|c||c|c|c|}\hline
D & \chi_{\rho}\geq & \chi_{\rho}\leq & \text{period} \\\hline\hline
1,2 & 8 & 8 & 54\\\hline
1,3 & 9^* & 9 & 32\\\hline
1,4 & 11 & 16 & 320\\\hline
1,5 & 10^* & 12 & 1028\\\hline
1,6 & 12 & 23 & 2016\\\hline
1,7 & 10^* & 15 & 640\\\hline
1,8 & 11^* & 25 & 5184\\\hline
1,9 & 10^* & 18 & 576\\\hline
1,2,3 & 17 & 23 & 768\\\hline
2,3 & 11 & 13 & 240\\\hline
2,5 & 14 & 23 & 336\\\hline
\end{array}$$
\caption{\label{tb1}Lower and upper bounds for the packing chromatic number of $G(\mathbb Z,D)$ for different values of
$D$. In the fourth column are the periods of the colorings giving the upper bounds. ($^*$: bound obtained by running Algorithm~\ref{algo1} of Section~\ref{ssmall}).}
\end{table}

The bounds of Section~\ref{s3} are summarized in the following Proposition:
\begin{prop}\label{p0} Let $t,q$ be integers. Then,
 $$\chi_{\rho}(D(1,t))\leq \left \{
\begin{array}{ll}
89, & t=2q+1, q\geq 35;\\
40, & t=2q+1,q\geq 223;\\
179, & t=2q,q\geq 89;\\
81, & t=2q,q\geq 224;\\
29, & t=96q\pm 1,q\geq 1;\\
59, & t=96q+1\pm 1,q\geq 1.\\
\end{array}\right.$$
\end{prop}

Some proofs of lower bounds use a density argument. For this, we define the density $\rho_a(G_{n}(D))$ of a color $a$ in
$G_{n}(D)$  as the maximum fraction of vertices colored $a$ in any packing coloring of $G_{n}(D)$ and $\rho_a(D)$ (or
simply $\rho_a$, if the graph is clear from the context) by $\displaystyle
\rho_a(D)=\limsup_{n\to+\infty}\rho_a(G_{n}(D))$. Let also $\rho_{1,2}(G_{n}(D))$ be the  maximum fraction of
vertices colored $1$ or $2$ in any packing coloring of $G_{n}(D)$ and let
$\displaystyle\rho_{1,2}=\limsup_{n\to+\infty}\rho_{1,2}(G_{n}(D))$.
We have trivially, for any $D$, $\chi_{\rho}(G(\mathbb Z,D)) \ge \min\{c \ | \sum_{i=1}^c \rho_i\ge 1\}$ and
$\rho_{1,2}\le \rho_1 + \rho_2$.

\section{Path Powers}\label{s2}
The $t^{th}$ power $G^t$ of a graph $G$ is the graph with the same vertex set as $G$ and edges between every vertices $x,y$ that are at a mutual distance of at most $t$ in $G$.
Let $D^t=G(\mathbb Z,\{1,2,\cdots ,t\})$ be the $t^{th}$ power of the two-ways infinite path and  let
$P_n^t=G_n(\{1,2,\cdots ,t\})$ be the $t^{th}$ power of the path $P_n$ on $n$ vertices.

We first present an asymptotic result on the packing chromatic number:

\begin{prop}  $\chi_{\rho}(D^t)=(1+o (1))3^t$ and 
$\chi_{\rho}(D^t)=\Omega(e^t )$.
\end{prop}

\begin{proof}
$D^t$ is a spanning subgraph of the lexicographic product\footnote{the lexicographic product $G\circ H$ of graphs $G$ and $H$ has vertex set $V(G)\times V(H)$ and two vertices $(a,x)$ and $(b,y)$ are linked by an edge if and only if $ab\in E(G)$ or $a=b$ and $xy\in E(H)$} $\mathbb{Z}\circ K_t$ (see Figure~\ref{fPt}). Then, as
Goddard et al.~\cite{GoBro} showed that $\chi_{\rho}(\mathbb{Z}\circ K_t )=(1+o(1))3^t$, the same upper
bound holds for $D^t$.
To prove the lower bound, since $\rho_i\leq \frac{1}{it+1}$, then for any packing coloring
of $D^t$ using at most $c$ colors, $c$ must satisfy:

$$\sum_{i=1}^c \frac{1}{it+1} \geq 1.$$

Since $$\sum_{i=1}^c \frac{1}{it+1} < \sum_{i=1}^c \frac{1}{it}=\frac 1t \sum_{i=1}^c \frac{1}{i} = \frac{H_c}{t},$$
where $H_n$ is the $n^{th}$ harmonic number and since $H_n=\Omega(\ln (n))$, then $\frac{H_c}{t} \geq 1$ implies $c=\Omega
(e^t )$.

\end{proof}

\begin{figure}[ht]
\begin{center}
\includegraphics[width=8cm]{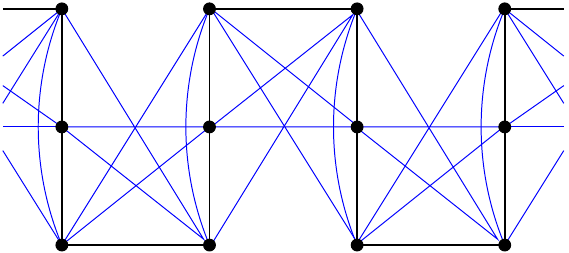}
\end{center}
\caption{\label{fPt} The infinite distance graph $D^3$ as a subgraph of the lexicographic product $\mathbb{Z} \circ
K_3$.}
\end{figure}

\begin{coro}
 For any finite subset $D$ of $\mathbb{N}$, the packing chromatic number of $G(\mathbb{Z},D)$ is finite.
\end{coro}

For very small $t$, exact values or sharp bounds for the packing chromatic number can be calculated:

\begin{prop}
 $$\chi_{\rho}(D^2) =8.$$
\end{prop}

\begin{proof}
 
A packing 8-coloring can be constructed by repeating the following pattern of length $54$ :
\begin{center}
$8,1,2,6,1,4,3,2,1,5,7,1,2,3,4,1,6,2,1,8,3,1,2,4,1,5,7,$\\
$1,3,2,1,6,4,1,2,3,1,8,5,1,2,4,1,3,6,1,2,7,1,5,4,2,1,3.$ \end{center}

On the other hand, it can be seen that $\rho_i \le \frac{1}{2i+1}$ for any $i\ge 1$. However, we next prove that $\rho_{1,2}\le \frac12$.
Consider vertices $v,v+1,\ldots, v+9$ for some $v$. The only possibility to color more than 5 of these $10$ vertices is
to give color 1 to $v, v+3, v+6, v+9$ and then at most 2 vertices can be given color 2 ($v+1$ or $v+2$, and $v+7$ or
$v+8$). But in this case, neither vertex $v+10$ nor vertex $v+11$ can be given color 1 or 2, resulting in $6$ vertices
colored out of $12$. 
Moreover, an easy computation gives that $\chi_{\rho}(D^2)\geq \min \{c\ |\ \frac{1}{2}+ \sum_{i=3}^c
\frac{1}{2i+1}\geq 1\}=8.$ 
\end{proof}

\begin{prop}
 $$17 \leq \chi_{\rho}(D^3)\leq 23.$$
\end{prop}

\begin{proof}
The upper bound comes from a packing $23$-coloring of period $768$ defined by repeating the sequence of length $768$ given in Appendix~\ref{app}.

To prove the lower bound, as the distance $\text{dist}(u,v)$ between the vertices $u$ and $v$ is $\text{dist}(u,v)= \lceil \frac{v-u}{3}\rceil$, then $\rho_i \le \frac{1}{3i+1}$ and an easy computation gives that $\chi_{\rho}(D^3)\geq \min \{c|\ \sum_{i=1}^c \frac{1}{3i+1}\geq 1\}=17.$
\end{proof}

\section{$D(1,t)$ with large $t$}
\label{s3}


The general method is to cut the distance graph into sets of consecutive vertices of size $s=t-1$ or
$s=t+1$, depending on the value of $t$ 
and to color each set by a predefined color pattern. 
\begin{figure}[ht]
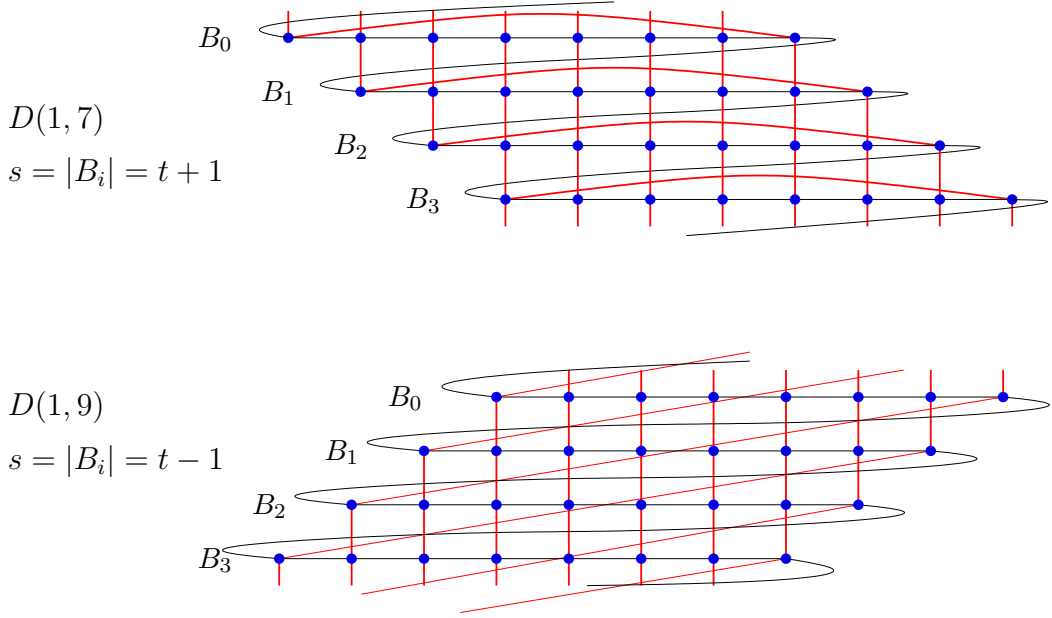

\begin{center}
\include{d1tGridTer}
\end{center}
\caption{\label{grid} $D(1,t)$, with $t=7$ (on the top) and $t=9$ (on the bottom) drawn by rows of size $s=8$.}
\end{figure}
Let $s$ be either $t+ 1$ or $t-1$ and let $A_i=\{is,is+1,\ldots, (i+1)s-1\}$ and $B_i$ be the subgraph of $D(1,t)$ induced by $A_i$. Notice that $V(D(1,t))=\bigcup_{i=-\infty}^{+\infty} A_i$ and that if $s=t+1$, then each $B_i$ is an induced cycle of $D(1,t)$ of length $s=t+1$ (see Figure~\ref{grid}).
By a {\em color pattern} $P$, we mean a sequence of integers $(c_1, c_2, \ldots , c_s )$ of length $s$ that will be associated to some subgraph $B_i$ by giving the color $c_j$ to the $j^{th}$ vertex of $B_i$. 
If $S$ is a sequence of integers, $S^p$ is the sequence obtained by repeating $S$ $p$ times. 
The {\em cyclic distance} between elements $s_i$ and $s_j$ of a sequence $(s_1, s_2, \ldots , s_\ell )$ is $\min(|j-i|, \ell-|j-i|)$.

We first need to know the distance between two vertices in $D(1,t)$.

\begin{Lemma}
\label{le1}
The distance between two vertices $u$ and $v$ of $D(1,t)$ is $\text{dist}(u,v)=\min (q+r,q+1+t-r)$, where $|v-u|=qt+r$, with $0\leq r<t$.
\end{Lemma}


\begin{proof}
Let us call an edge joining vertices $x$ and $y$, with $|y-x|=k$ a $k$-edge. Assume, without loss of generality, that $v\geq u$. then, any minimal path between $u$ and $v$ uses either $q$ $t$-edges and $r$ $1$-edges or $q+1$ $t$-edges and $t-r$ $1$-edges.
\end{proof}

The key lemma of our method is the following one which gives conditions for a coloring of $D(1,t)$ by color patterns to
be a packing coloring.

\begin{Lemma}
\label{le2}
Let $s>1$ be a positive integer and for each integer $i$, set $A_i=\{is, is+1, \ldots, (i+1)s-1\}$. Let $t$ be a positive integer and for each $i$, let $B_i$ be the subgraph of $G=D(1,t)$ induced by $A_i$, and $C_i$ be the graph $B_i$ with an additional edge joining vertices $is$ and $(i+1)s-1$ if $s=t-1$. Suppose that $G$ is colored in such a way that:
\begin{description}
 \item[i)] for each integer $i$, the coloring inherited by each $C_i$ is a packing coloring;
 \item[ii)] for each pair of integers $i$ and $j$, if $c$ is the maximum common color in both $C_i$ and $C_j$ then we
have $c< s$, $|i-j|>\frac{c}{2}$, and each $b\le c$ that is a common color in both $C_i$ and $C_j$ has the property that
$si+k$ is colored $b$ if and only if $sj+k$ is colored $b$ for each $k\in \{0,1,\ldots,s-1\}$.
\end{description}
Then the coloring is a packing coloring of $G$ whenever $t$ is in $\{s+1,s-1\}$.
\end{Lemma}

\begin{proof} 
Suppose vertices $u$ and $v$ have the same color, say $e$, and, without loss of generality, assume $u$ is in $B_0$. Let $\sigma : V(G) \rightarrow V(C_0)$ be defined by $\sigma(k)=k\bmod{s}$ for each $k\in \mathbb{N}$. Observe that when $t=s+1$ or $t=s-1$, if two vertices $x$ and $y$ are adjacent in $G$, then $\sigma(x)$ and $\sigma(y)$ are adjacent in $C_0$. But then a path in $G$ between $u$ and $v$ maps via $\sigma$ to a path of at most the same length between two vertices in $C_0$ colored $e$. Since, by hypothesis, $C_0$ is colored by a packing coloring, as long as $u\ne \sigma(v)$, the distance between $u$ and $v$ must be greater than $e$.

If $u=\sigma(u)=\sigma(v)$, then $v-u=js$ for some $j$. 
If $s=t-1$, then $v-u=j(t-1)=(j-1)t +t-j$ and by Lemma~\ref{le1}, $\text{dist}(u,v)=\min(j-1+t-j,j+t-t+j)=\min(t-1,2j)>e$ since by hypothesis, $e < s=t-1$ and $2j>e$.
Else, if $s=t+1$ then $v-u=j(t+1)=jt+j$ and by Lemma~\ref{le1}, $\text{dist}(u,v)=\min(j+j,j+1+t-j)=\min(2j,t+1)>e$ by
hypothesis.
\end{proof}

\subsection{Proof of Proposition~\ref{p0}}

\begin{proof} 
Let $t$ be an integer, $G=D(1,t)$ and $s=4p$ if $t=4p-1$ or $t=4p+1$ for some $p$; $s=4p+1$ if $t=4p$ or $t=4p+2$.
For each integer $i$, set $A_i=\{is, is+1, \ldots, (i+1)s-1\}$ and let $B_i$ be the subgraph of $G$ induced by $A_i$.

In each of the following cases, a packing coloring of $G$ is defined by assigning to each subgraph $B_i$ a pattern of colors with length $s$.
We will use the following sub-patterns of colors:

\noindent$S_{2,3}\ =(1,2,1,3)$,\\
$S_{4,9}\ =(1,4,1,5, 1, 8, 1, 4, 1, 5, 1, 9)$,\\
$S_{4,11}=(1,4,1,5, 1, 10, 1, 4, 1, 5, 1, 11)$,\\
$S_{6,15}=(1, 6, 1, 7, 1, 12, 1, 13, 1, 6, 1, 7, 1, 14, 1, 15)$,\\
$S_{6,21}=(1, 6, 1, 7, 1, 16, 1, 17, 1, 6, 1, 7, 1, 18, 1, 19, 1, 6, 1, 7, 1, 20, 1, 21)$,\\
$S_{6,29}=(1, 6, 1, 7, 1, 22, 1, 23, 1, 6, 1, 7, 1, 24, 1, 25, 1, 6, 1, 7, 1, 26, 1, 27, 1, 6, 1, 7, 1, 28, 1, 29)$,\\
$S_{6,31}=(1,6,1, 7,1, 22,1, 23,1, 6,1, 7,1, 24,1, 25,1, 6,1, 7,1, 26,1, 27,1,6,1,7,1,22,1,23,$\\
\hspace*{1.3cm}$1,6,1,7,1,28,1, 29,1,6,1,7,1,30,1,31)$.

%

\paragraph{Case A. $t$ is odd.} 

First, since $s=4p$ for some integer $p$ and thanks to Lemma~\ref{le2}, we can assign to each subgraph $B_{2i+1}$ the color pattern $(S_{2,3})^{p}$.
In order to color subgraphs $B_{2i}$, we consider three sub-cases (that are not totally disjoints).

\paragraph{Subcase A.1. $t=96q\pm 1$ for some $q\geq 1$.}

A packing coloring of $D(1,t)$ using these sub-patterns is constructed by assigning inductively to $8$ consecutive subgraphs $B_{2i}$ 
the sequence of color patterns 
$$\mathcal{P}=((S_{4,9})^{8q},(S_{6,15})^{6q}, (S_{4,11})^{8q},(S_{6,21})^{4q},(S_{4,9})^{8q},(S_{6,15})^{6q}, (S_{4,11})^{8q},(S_{6,29})^{3q}).$$

Since the cyclic distance between two occurrences of any color $e$ in each color pattern is always greater than $e$,
then Condition i) of Lemma~\ref{le2} is satisfied. Moreover, as the cyclic distance between any two color patterns in
$\mathcal{P}$ is always greater than a quarter (since color patterns of $\mathcal{P}$ are associated only with subgraphs of even indices) of their maximum common color, then Condition ii) is also satisfied. Hence, the coloring is a packing coloring
of $D(1,t)$ and $\chi_{\rho}(D(1,t))\le 29$.

\paragraph{Subcase A.2. $t=2p+1$ for some $p\geq 223$.}

We denote by $S\shu (1,\alpha)^r$ any sequence obtained by inserting $r$ quasi evenly cyclically-distributed occurrences of
the pair $(1,\alpha)$ in the sequence $S$; insertions being made only after a color different from $1$, in order to keep the sequence alternate between color 1 and other colors.\\
For example, $(1,4,1,5,1,8,1,4,1,5,1,9)^3 \shu (1,\alpha)^5$ can be rewritten as\\
{\small $(1,4,1,5,1,8,\boldsymbol{1,\alpha,}1,4,1,5,1,9,\boldsymbol{1,\alpha,}1,4,1,5,1,8,1,4,\boldsymbol{1,\alpha,}1,5,1,9,1,4,1,5,\boldsymbol{1,\alpha,}
1,8,1,4,1,5,1,9,\boldsymbol{1,\alpha})$.}

Then, color patterns using colors from $\{1,2,\ldots,40\}$ are defined by:

\noindent
$Q^i_{1}$ $= (S_{4,9})^{q_1}\shu (1,32+i)^{r_1}$, for $s=12q_1 +2r_1, 0\leq r_1 \leq 4$, $i=0,1,2$;\\
$Q^i_{2}$ $= (S_{4,11})^{q_2}\shu (1,35+i)^{r_2}$, for $s=12q_2 +2r_2, 0\leq r_2 \leq 4$, $i=0,1,2$;\\
$Q^i_{3}$  $= (S_{6,15})^{q_3}\shu (1,38+i)^{r_3}$, for $s=16q_3 + 2r_3, 0\leq r_3 \leq 6$, $i=0,1,2$;\\
$Q_{4}$  $= (S_{6,21})^{q_4}\shu(1,30)^{r_4}$, for $s=24q_4 + 2r_4, 0\leq r_4 \leq 10$;\\
$Q_{5}$  $= (S_{6,29})^{q_5}\shu(1,31)^{r_5}$, for $s=32q_5 + 2r_5, 0\leq r_5 \leq 14$;

and we assign inductively to $24$ consecutive subgraphs $B_{2i}$ the sequence of color patterns $\mathcal{Q}$ defined by
$$\mathcal{Q}= (Q^0_{1},Q^0_{3}, Q^0_{2},Q_{4},Q^1_{1},Q^1_{3},Q^1_{2},Q_{5},Q^2_{1},Q^2_{3},
Q^2_{2},Q_{4},Q^0_{1},Q^0_{3}, Q^0_{2},Q_{5},Q^1_{1},Q^1_{3}, Q^1_{2},Q_{4},Q^2_{1},Q^2_{3}, Q^2_{2},Q_{5}).$$

 In order for a color pattern $S\shu (1,\alpha)^r$ to satisfy Condition $i)$ of Lemma~\ref{le2} and as the pairs $(1,\alpha)$
have to be inserted only on even positions, we must have $2\lfloor\lfloor\frac{|S|}{r}\rfloor/2\rfloor\geq \alpha$.
Hence the worst case for this separation constraint is for color $31$ in $Q_{5}$ when $r_5=14$: one can insert
$14$ occurrences of $(1,31)$ if $2\lfloor\lfloor\frac{32q_5}{14}\rfloor/2\rfloor\geq 31$, which is true as soon as $q_5=14$ and thus $s=448$.
Moreover, it can be seen that the added color in each pattern is chosen in such a way that Condition $ii)$ is satisfied.
Hence, the coloring is a packing coloring of $D(1,t)$ and $\chi_{\rho}(D(1,t))\le 40$.

\paragraph{Subcase A.3. $t=2p+1$ for some $p$, $35 \leq p\leq 222$.}
The base case is $s\equiv 0\pmod{48}$ for which the sequence of color patterns that is assigned inductively to $8$ consecutive subgraphs $B_{2i}$ is defined as follows:
$$\mathcal{R}=(R_{1}, R_{3}, R_{2}, R_{4}, R_{1}, R_{3}, R_{2}, R_{5}),$$ with $R_{1}=
(S_{4,9})^{4q}$, $R_{2}=(S_{4,11})^{4q}$, $R_{3}= (S_{6,15})^{3q}$, $R_{4}= (S_{6,21})^{2q}$, and $R_{5}=
(S_{6,31})^{q}$.

As for Subcase A.1, it can be easily checked that the defined coloring is a packing coloring.

Now, for $s\not\equiv 0\pmod{48}$, we may replace each of the above color patterns $R_j\in \mathcal{R}$ by a certain
number of patterns $R^i_j$ (depending on the residue of $s$ modulo the length of the sub-pattern used) that will be used
in turn, as for Subcase A.2. 

Let $\epsilon$ be the empty sequence and let $c_j$ and $\delta_j$, $1\le j\le 5$ be some integers (that will be set just after).

Set $R^i_{1}=(S_{4,9})^{q_1}.T^i_1$, with $s=12q_1 + 4r_1$, $0\leq r_1 < 3$, $0\leq i< \delta_1$, and\\
$T^i_1=\left\{\begin{array}{ll}
\epsilon, & \text{ if } r_1=0;\\
(1,c_1+i,1,c_1+\delta_1 +i), & \text{ if } r_1=1;\\
(1,4,1,5,1,c_1+i,1,c_1+\delta_1 +i), & \text{ if } r_1=2.\\
\end{array}\right.$\\

Set $R^i_{2}=(S_{4,11})^{q_2}.T^i_2$, with $s=12q_2 + 4r_2$, $0\leq r_2 < 3$, $0\leq i < \delta_2$, and \\
$T^i_2=\left\{\begin{array}{ll}
\epsilon, & \text{ if } r_2=0;\\
(1,c_2+i,1,c_2+\delta_2 +i), & \text{ if } r_2=1;\\
(1,4,1,5,1,c_2+i,1,c_2+\delta_2 +i), & \text{ if } r_2=2.\\
\end{array}\right.$

Set $R^i_{3}=(S_{6,15})^{q_3}.T^i_3$, with $s=16q_3 + 4r_3$, $0\leq r_3 < 4$, $0\leq i < \delta_3$, and\\
$T^i_3=\left\{\begin{array}{ll}
\epsilon, & \text{ if } r_3=0;\\
(1,c_3+i,1,c_3+\delta_3 +i), & \text{ if } r_3=1;\\
(1,6,1,7,1,c_3+i,1,c_3+\delta_3 +i), & \text{ if } r_3=2;\\
(1,6,1,7,1,c_3+i,1,c_3+\delta_3 +i,1,c_3+2\delta_3 +i,1,c_3+3\delta_3 +i), & \text{ if } r_3=3.\\
\end{array}\right.$

\medskip
Set $R^i_{4}=(S_{6,21})^{q_4}.T^i_4$, with $s=24q_4 + 4r_4$, $0\leq r_4 < 6$, $0\leq i < \delta_4$, and \\
$T^i_4=\left\{\begin{array}{ll}
\epsilon, & \text{ if } r_4=0;\\
(1,c_4+i,1,c_4+\delta_4 +i), & \text{ if } r_4=1;\\
(1,6,1,7,1,c_4+i,1,c_4+\delta_4 +i), & \text{ if } r_4=2;\\
(1,6,1,7,1,c_4+i,1,c_4+\delta_4 +i,1,c_4+2\delta_4 +i,1,c_4+3\delta_4 +i), & \text{ if } r_4=3;\\
(1,6,1,7,1,c_4+i,1,c_4+\delta_4 +i,1,6,1,7,1,c_4+2\delta_4 +i,1,c_4+3\delta_4 +i), & \text{ if } r_4=4;\\
(1,6,1,7,1,c_4+i,1,c_4+\delta_4 +i,1,6,1,7,1,c_4+2\delta_4 +i,1,c_4+3\delta_4 +i,& \\
\hspace*{.2cm}1,c_4+4\delta_4 +i,1,c_4+5\delta_4 +i), & \text{ if } r_4=5;\\
\end{array}\right.$

\medskip
Set $R^i_{5}=(S_{6,31})^{q_5-1}.T^i_5$, with $s=48q_5 + 4r_5$, $0\leq r_5 < 12$, $0\leq i < \delta_5$, and\\
$T^i_5=\left\{\begin{array}{ll}
S_{6,31}, & \text{ if } r_5=0;\\
S_{6,31}.(1,c_5+i,1,c_5+\delta_5 +i), & \text{ if } r_5=1;\\
S_{6,31}.(1,6,1,7,1,c_5+i,1,c_5+\delta_5 +i), & \text{ if } r_5=2;\\
S_{6,31}.(1,6,1,7,1,c_5+i,1,c_5+\delta_5 +i,1,c_5+2\delta_5 +i,1,c_5+3\delta_5 +i), & \text{ if } r_5=3;\\
(S_{6,29})^2, & \text{ if } r_5=4;\\
(S_{6,29})^2 .(1,c_5+i,1,c_5+\delta_5 +i), & \text{ if } r_5=5;\\
(S_{6,29})^2 .(1,6,1,7,1,c_5+i,1,c_5+\delta_5 +i), & \text{ if } r_5=6;\\
(S_{6,29})^2 .(1,6,1,7,1,c_5+i,1,c_5+\delta_5 +i,1,c_5+2\delta_5 +i,1,c_5+3\delta_5 +i), & \text{ if } r_5=7;\\
S_{6,31}.S_{6,29}, & \text{ if } r_5=8;\\
S_{6,31}.S_{6,29} .(1,c_5+i,1,c_5+\delta_5 +i), & \text{ if } r_5=9;\\
S_{6,31}.S_{6,29} .(1,6,1,7,1,c_5+i,1,c_5+\delta_5 +i), & \text{ if } r_5=10;\\
S_{6,31}.S_{6,29} .(1,6,1,7,1,c_5+i,1,c_5+\delta_5 +i,1,c_5+2\delta_5 +i,1,c_5+3\delta_5 +i), & \text{ if } r_5=11;\\
\end{array}\right.$

\medskip

As the cyclic distance between two occurrences of either the color pattern $R_{1}$ or of $R_{2}$ or of $R_{3}$ in
$\mathcal{R}$ is equal to 4 (hence, each of these three patterns appears every 8 set $B_i$), and if $e$ is the maximum color used in $R^i_j$, then, according to Lemma 2,
for $j=1,2,3$, $\delta_j$ must satisfy
$$\delta_j\ge \left \{ \begin{array}{l}
                    1, \text{ if } e \le 15;\\
		    2, \text{ if } 16\le e \le 31;\\
		    3,  \text{ if } 32 \le e \le 47;\\
                    4,  \text{ if } 48 \le e \le 63;\\
		    5,  \text{ if } 64 \le e \le 79.
                   \end{array}\right.$$ 

Similarly, the cyclic distance between two occurrences of either the color pattern $R_{4}$ or of $R_{5}$ in
$\mathcal{R}$ is equal to 8, hence, for $j=4$ or $5$, $\delta_j$ must satisfy
 $$\delta_j\ge \left \{ \begin{array}{l}
                    1, \text{ if } e \le 31;\\
		    2, \text{ if } 32 \le e \le 63;\\
		    3,  \text{ if } 64 \le e \le 95.
                   \end{array}\right.$$

Therefore, for each residue of $s$ modulo $48$, a packing coloring is obtained by fixing the values of $c_j$ and $\delta_j$ as
indicated in the next table ($\delta_j$ is set to the smallest value satisfying the above inequations). The largest color used
in each case is reported on the last row.
{\small
\begin{center}\begin{tabular}{|c|c|c|c|c|c|c|c|c|c|c|c|c|}\hline
$s\pmod{48}$ 		& 0       & 4	    & 8		& 12	  & 16	     & 20	& 24	   & 28       & 32 	& 36 	   & 40       & 44\\\hline\hline
$c_1$, $\delta_1$	& /       & $32$, 3 & $32$, 3	& /   	  & $32$, 3  & $32$, 3	& /   	   & $32$, 3  & $32$, 3 &/         & $32$, 3  & $32$, 3\\\hline
$c_2$, $\delta_2$	& /       & $38$, 3 & $38$, 3	& /   	  & $38$, 3  & $38$, 3	& /   	   & $38$, 3  & $38$, 3 & /   	   & $38$, 3  & $38$, 3 \\\hline
$c_3$, $\delta_3$ 	& /       & $44$, 4 & $44$, 4	& $32$, 3 & /  	     & $44$, 4	& $32$, 3  & $44$, 4  & /       & $32$, 3  & $44$, 4  & $44$, 4 \\\hline
$c_4$, $\delta_4$ 	& /       & $52$, 2 & $52$, 2	& $44$, 2 & $44$, 2  & $52$, 2	& /   	   & $60$, 2  & $44$, 2 & $38$, 2  & $52$, 2  & $60$, 2 \\\hline
$c_5$, $\delta_5$ 	& /       & $56$, 2 & $56$, 2	& $52$, 2 & /        & $64$, 3	& $38$, 2  & $64$, 3  & /    	& $46$, 2  & $60$, 2  & $78$, 2
\\\hline\hline
largest color		& 31      & 59	    & 59	& 59	  & 51	     & 69	& 41	   & 75	      & 47	& 49	   & 63	      & 89\\\hline
\end{tabular}
\end{center}
}
 
An illustration for the case $s\equiv 28 \pmod{48}$ is given in Appendix~\ref{a2}.
 
\paragraph{Case B. $t$ is even.}
For $t=4p$ or $t=4p+2$, recall that subgraphs $B_i$ are of size $s=4p+1$. New color patterns are constructed by
inserting a  new color at the end of each pattern (of length $s'=s-1=4p$) defined in Subcases A.1, A.2 and A.3.

By Lemma~\ref{le2}, the problem of adding the missing color in each color pattern defined in subcases A.1, A.2
and A.3 is equivalent to the one of coloring the infinite path $P_{\infty}$ with colors from $\{k_1,k_1 +1,\ldots,
k_2\}$ 
such that vertices of color $e$ are at distance greater than $\frac{e}{2}$.

We are going to show, by induction on $k_1$, that $k_2 \leq 2k_1 -1$. For $k_1=2$, vertices can be colored by alternating color $2$ and color $3$, 
so $k_2=3$. Assume that $P_{\infty}$ can be colored with colors from $\{k_1,k_1 +1,\ldots, k_2\leq 2k_1 -1\}$ and let $k'_1 =k_1 +1$. 
Replace now color $k_1$ by colors $k_2+1$ and $k_2+2$ alternatively. Then the largest color used is $k'_2 = k_2+2\leq
2k_1 +1 = 2k'_1 -1$ 
and the constraint is satisfied since if vertices $x$ and $y$ are colored $k_2+2$ then their mutual distance satisfies 
$\text{dist}(x,y)> 2\frac{k_1}{2}\geq \frac{k_2+1}{2}>\frac{k_2}{2}$.

As the colorings defined in Subcase A.1 (Subcases A.2 and A.3, respectively) use colors from $1$ to $29$ ($40$ and at most $89$,
respectively), then we obtain a packing coloring of $D(1,t)$ with colors from $1$ to at most $2\times 30-1=59$ ($81$ and
$179$, respectively), provided that $t\geq 96$ ($448$ and $144$, respectively).

\end{proof}

\begin{Remark}

\begin{itemize}
\item In Subcase A.2, the method can produce a packing coloring using less than $40$ colors, depending on the value of
$s$ (i.e. if some $r_i$ are equal to zero). 
\item A combination of the methods of Subcases A.2 and A.3 could be used to define a packing coloring for odd $t$, $95 \leq t\leq
447$, using less colors than in Subcase A.3.
 
 \item For Case B, it seems that less than  $2k_1 -1$ colors are sufficient for such a coloring. When $k_1 = 90$, a
computation gives $k_2=156$ for such a coloring; when $k_1 = 41$, we find $k_2=72$ and when $k_1=30$, we find $k_2=53$.
\end{itemize}
\end{Remark}

\section{$D(a,b)$ with small $a$ and $b$}
\label{ssmall}
The results from Section~\ref{s3} do not apply for $D(1,t)$ with small $t$, however it is possible to derive exact or sharp 
results for some of them, using density arguments and the computer.

Algorithm~\ref{algo1} is a simple algorithm that prints all the packing $k$-colorings of $D_n(1,t)$. It checks, for each vertex, each possible color in a recursive fashion. Hence it must be used by initializing the first $n$ elements of the array {\tt color} to $0$ and calling $\text{RecColor}(0)$.

\begin{algorithm}[ht]
\label{algo1}
\KwData{global integers $n, k, t$; global array {\tt color};}

 \If{$i=n$}{
    print({\tt color})\;
     }
 \Else{ 
 \For{$c$ from 1 \KwTo $k$}{
  \If{$ \nexists$ $j < i$ such that {\tt color}[$j$]=$c$ and $dist(i,j)\le c$}
    { {\tt color}[$i$] $\leftarrow$ $c$\;
  RecColor($i+1$)\;
  {\tt color}[$i$] $\leftarrow$ 0\;
     }
  }
}

\caption{RecColor($i$)}
\end{algorithm}

\begin{prop}
 $$\chi_{\rho}(D(1,3)) = 9.$$
\end{prop}

\begin{proof}
first, remark that the graph-distance $\text{dist}(i,j)$ between vertex $i$ and vertex $j\geq i$ is $\text{dist}(i,j)=\lfloor \frac{j-i}{3}\rfloor + (j-i) \pmod{3}$.

A packing 9-coloring of $D(1,3)$ of period $32$ is given by the following sequence:
$$1,2,1,3,1,4,1,5,1,2,1,3,1,6,1,7,1,2,1,3,1,4,1,5,1,2,1,3,1,8,1,9.$$
It is routine to check that the vertices of a same color are sufficiently distant.
On the other hand, running an implementation of Algorithm~\ref{algo1} with $n=100$, $k=8$, and $t=3$, outputs no coloring, showing that $8$ colors are not sufficient for a packing coloring of $D_{100}(1,3)$.


\end{proof}

\begin{prop}

$$11\leq \chi_{\rho}(D(1,4)) \leq 16;$$
$$10\leq \chi_{\rho}(D(1,5)) \leq 12;$$
$$12\leq \chi_{\rho}(D(1,6)) \leq 23;$$
$$10\leq \chi_{\rho}(D(1,7)) \leq 15;$$
$$11\leq \chi_{\rho}(D(1,8)) \leq 25;$$
$$10\leq \chi_{\rho}(D(1,9)) \leq 18.$$

\end{prop}

\begin{proof} 
For the upper bounds, packing $k$-colorings are defined by exhibiting a pattern using colors from $\{1,\cdots, k\}$ of length $\ell$ for each case.
For $D(1,4)$, the pattern with $k=16$ and $\ell=320$ is given in Appendix~\ref{app}.
For $D(1,5)$ ($D(1,6)$, $D(1,7)$, $D(1,8)$, and $D(1,9)$, respectively), the pattern with $(k,\ell)=(12,1028)$ ($(23,2016)$, $(15,640)$, $(25,5184)$, $(18,576)$, respectively) can be found at {\tt http://www.u-bourgogne.fr/o.togni/PCDG.html}.

For the lower bounds, we use either density arguments or computer running Algorithm~\ref{algo1}.

For $D(1,4)$, we have $\rho_1 \le \frac25$ since at most 2 out of 5 consecutive vertices can be colored 1. Moreover,
$\rho_i\le\frac{1}{4i-2}$ for $i\geq 2$ and $\min \{c |\ \frac{2}{5}+ \sum_{i=2}^c \frac{1}{4i-2}\geq 1\}=11$.

For $D(1,5)$, running an implementation of Algorithm~\ref{algo1} with $n=43$, $k=9$, and $t=5$ outputs no coloring. Hence $\chi_{\rho}(D(1,5))\ge 10$.

For $D(1,6)$, we have $\rho_1 \le \frac37$ since at most 3 out of 7 consecutive vertices can be colored 1. We now show
that $\rho_2 \le \frac{2}{11}$. Let $v$ be a vertex colored 2. If $v+3$ is also colored 2, then no vertex of
$\{v+4,\cdots ,v+10\}$ can be colored 2. Hence 2 vertices out of 11 are colored 2.
If $v+3$ is not colored 2 but $v+4$ is, then only one of $v+8$, $v+14$ can be colored 2 among $\{v+5,\cdots, v+16\}$,
resulting in 3 out of 17 vertices colored 2 and $\frac{3}{17} < \frac{2}{11}$.
If neither $v+3$ nor $v+4$ is colored 2 then no vertex of $\{v+5,v+6,v+7\}$ can be colored 2 and at most one vertex of
$\{v+8,v+9,v+10\}$ can have color 2, resulting in 2 vertices out of 11 colored 2.
Moreover, if $i\ge 3$, then $\rho_i \le \frac{1}{6i-9}$ and $\min \{c |\ \frac{3}{7}+ \frac{2}{11} + \sum_{i=3}^c
\frac{1}{6i-9}\geq 1\}=12$.

For $D(1,7)$, running an implementation of Algorithm~\ref{algo1} with $n=44$, $k=9$, and $t=7$ outputs no coloring. Hence $\chi_{\rho}(D(1,7))\ge 10$.

For $D(1,8)$, running an implementation of Algorithm~\ref{algo1} with $n=41$, $k=10$, and $t=8$ outputs no coloring. Hence $\chi_{\rho}(D(1,8))\ge 11$.

For $D(1,9)$, running an implementation of Algorithm~\ref{algo1} with $n=46$, $k=9$, and $t=9$ outputs no coloring. Hence $\chi_{\rho}(D(1,9))\ge 10$.
\end{proof}

It is interesting to notice that sometimes adding just one more color allows us to shorten considerably the period of the packing coloring, 
as can be seen with $D(1,5)$ with the following periodic packing $13$-coloring
of period $80$ (compared with the packing $12$-coloring of period $1028$):
\begin{center}$1,2,1,3,1,4,1,5,1,2,1,3,1,6,1,7,1,2,1,3,1,10,1,4,1,2,1,3,1,5,1,11,1,2,1,3,1,8,1,9,$\\
$1,2,1,3,1,4,1,5,1,2,1,3,1,6,1,7,1,2,1,3,1,12,1,4,1,2,1,3,1,5,1,13,1,2,1,3,1,9,1,8.$
\end{center}

\medskip

We now turn our attention to other $4$-regular distance graphs, i.e. graphs of type $D(a,b)$, with $2\leq a\leq b$.
First, remark that if $a$ and $b$ are not co-prime, then the graph $D(a,b)$ is not connected and consists in $g=\gcd(a,b)$ copies of $D(\frac ag ,\frac bg )$. Hence we only consider distance graphs $D(a,b)$ with $\gcd(a,b)=1$.

 The smallest example is $D(2,3)$ which is a subgraph of $D(1,2,3)=P_{\infty}^3$, thus 
$\chi_{\rho}(D(2,3))\leq \chi_{\rho}(P_{\infty}^3)\leq 23$. In fact, we show that the upper bound is much less than 23:

\begin{prop}
$$11\leq \chi_{\rho}(D(2,3)) \leq 13;$$
$$14\leq \chi_{\rho}(D(2,5)) \leq 23.$$
\end{prop}

\begin{proof}
The lower bound $11\leq \chi_{\rho}(D(2,3))$ is obtained by calculating the maximum density $\rho_i$ of a color $i$: 
it can be seen that $\rho_1 = \frac25$ and $\rho_i=\frac{1}{3i+1}$ for $i\geq 2$  and that 
$\min \{c |\ \frac{2}{5}+ \sum_{i=2}^c \frac{1}{3i+1}\geq 1\}=11$. 

For the lower bound $14\leq \chi_{\rho}(D(2,5))$, it can be seen that $\rho_1 = \frac37$ 
and $\rho_i=\frac{1}{5i-4}$ for $i\geq 2$  and that 
$\min \{c |\ \frac{2}{5}+ \sum_{i=2}^c \frac{1}{5i-4}\geq 1\}=14$. 

The upper bounds come from the packing $13$-coloring of $D(2,3)$ of period $240$ and the packing $23$-coloring of $D(2,5)$ of period $336$ given in Appendix~\ref{app}.
\end{proof}

\section{Concluding remarks}
\label{s4}
We have shown that the packing chromatic number of any infinite distance graph with finite $D$ is finite and is at most $40$ ($81$, respectively)
for $D=\{1,t\}$ with $t$ being an odd (even, respectively) integer greater than or equal to $447$.

Among the many possible research directions, one can try to find better bounds and/or more simple methods for
$D(1,t)$. 
In fact, running a simple greedy packing coloring algorithm that consists in coloring vertices
 of a distance graph one-by-one from the left to the right with the smallest color with respect to the constraint, suggests that
the upper bounds found in Section~\ref{s3} can be strengthened. Figure~\ref{fgreed} shows the number of colors used by
the greedy algorithm for a packing coloring of $D_n(1,t)$ (with $n=1000000$) as a function of $t$ for the first $500$ values of $t$.
\begin{figure}
 \includegraphics[width=13cm]{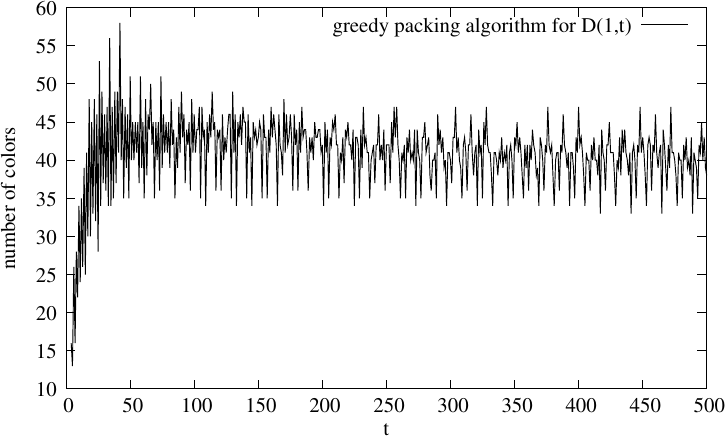}
\caption{\label{fgreed}Number of colors for a packing coloring of $D(1,t)$ using a greedy algorithm.}
\end{figure}
One  can see on the figure that for large $t$, the algorithm finds a packing coloring, using between $30$ and $50$
colors. 
Moreover, more colors are needed in general when $t$ is even compared to when $t$ is odd. 
But surprisingly, even if we look only at even (or odd) values of $t$, the function is not monotonic. We wonder if
the same goes for $\chi_{\rho}$.
An interesting future work would be to study in more details the behavior of this greedy algorithm.

Finally, a summary of the values of $t$ for which a upper bound on the the packing chromatic number of $D(1,t)$ is known and those that remain to be found is presented in Table~\ref{tb2}.

\begin{table}[h]

\begin{center}
\noindent\begin{tabular}{|c||c|c|c|c|c|}\hline
odd $t$ & $11\rightarrow45$ & $47,49$ & $51\rightarrow69$ & $71\rightarrow445$ & $447\rightarrow+\infty$ \\\hline\hline
$\chi_{\rho}\le$ & ? & 31 &  ? & between $29$ and $89$ &  40\\\hline\hline
%
even $t$ & $10\rightarrow94$ & $96,98$ & $100\rightarrow142$ & $144\rightarrow446$ & $448\rightarrow+\infty$ \\\hline\hline
$\chi_{\rho}\le$ & ? & 59 & ? & between $59$ and $179$ &  81\\\hline
\end{tabular}
\end{center}
 \caption{\label{tb2}Known upper bounds for the packing chromatic number of $D(1,t)$}
\end{table}

\section*{Acknowledgements}
The author is very indebted to the anonymous referees for their careful reading of the manuscript and their accurate comments; in particular to one of them for suggesting the statement and proof of Lemma 2 in its present form and many other changes. We wish also to thank P\v{r}emysl Holub for the valuable discussions and for his pertinent comments on a preliminary version of the paper.

\newpage

\appendix
\section{Periodic packing coloring of some distance graphs}
\label{app}
\noindent
\textbf{A periodic packing 23-coloring of $D(1,2,3)$ of period $768$}\\
\begin{footnotesize}
$23,1,4,5,3,1,2,6,7,1,9,10,12,1,2,3,4,1,8,5,13,1,2,14,16,1,3,6,11,1,2,4,7,1,15,5,3,1,2,9,18,1,10,8,4,$\\
$1,2,3,6,1,12,5,17,1,2,7,19,1,3,4,13,1,2,11,20,1,14,5,3,1,2,6,4,1,8,9,10,1,2,3,7,1,15,5,16,1,2,4,12,1,$\\
$3,6,21,1,2,18,22,1,11,5,3,1,2,4,7,1,8,9,10,1,2,3,6,1,13,5,4,1,2,14,17,1,3,19,23,1,2,7,12,1,4,5,3,1,2,$\\
$6,8,1,9,10,11,1,2,3,4,1,15,5,16,1,2,7,18,1,3,6,13,1,2,4,20,1,8,5,3,1,2,9,12,1,10,14,4,1,2,3,6,1,7,5,$\\
$11,1,2,17,19,1,3,4,8,1,2,21,15,1,22,5,3,1,2,6,4,1,7,9,10,1,2,3,12,1,13,5,16,1,2,4,8,1,3,6,11,1,2,14,$\\
$7,1,18,5,3,1,2,4,9,1,20,10,17,1,2,3,6,1,8,5,4,1,2,7,12,1,3,13,15,1,2,11,19,1,4,5,3,1,2,6,9,1,10,8,14,$\\
$1,2,3,4,1,7,5,16,1,2,21,22,1,3,6,18,1,2,4,12,1,11,5,3,1,2,8,7,1,9,10,4,1,2,3,6,1,13,5,15,1,2,14,17,1,$\\
$3,4,19,1,2,7,8,1,20,5,3,1,2,6,4,1,9,10,11,1,2,3,12,1,16,5,18,1,2,4,7,1,3,6,8,1,2,13,21,1,14,5,3,1,2,$\\
$4,9,1,10,15,17,1,2,3,6,1,7,5,4,1,2,8,11,1,3,12,19,1,2,20,22,1,4,5,3,1,2,6,7,1,9,10,13,1,2,3,4,1,8,5,$\\
$14,1,2,16,18,1,3,6,11,1,2,4,7,1,12,5,3,1,2,9,15,1,10,8,4,1,2,3,6,1,17,5,13,1,2,7,19,1,3,4,20,1,2,11,$\\
$14,1,21,5,3,1,2,6,4,1,8,9,10,1,2,3,7,1,12,5,16,1,2,4,15,1,3,6,13,1,2,18,22,1,11,5,3,1,2,4,7,1,8,9,10,$\\
$1,2,3,6,1,14,5,4,1,2,12,17,1,3,19,20,1,2,7,23,1,4,5,3,1,2,6,8,1,9,10,11,1,2,3,4,1,13,5,15,1,2,7,16,1,$\\
$3,6,12,1,2,4,14,1,8,5,3,1,2,9,18,1,10,21,4,1,2,3,6,1,7,5,11,1,2,17,19,1,3,4,8,1,2,13,20,1,12,5,3,1,2,$\\
$6,4,1,7,9,10,1,2,3,14,1,15,5,16,1,2,4,8,1,3,6,11,1,2,18,7,1,22,5,3,1,2,4,9,1,12,10,13,1,2,3,6,1,8,5,$\\
$4,1,2,7,17,1,3,14,19,1,2,11,15,1,4,5,3,1,2,6,9,1,10,8,16,1,2,3,4,1,7,5,12,1,2,13,18,1,3,6,20,1,2,4,21,$\\
$1,11,5,3,1,2,8,7,1,9,10,4,1,2,3,6,1,14,5,15,1,2,17,19,1,3,4,12,1,2,7,8,1,13,5,3,1,2,6,4,1,9,10,11,1,2,$\\
$3,16,1,18,5,22,1,2,4,7,1,3,6,8,1,2,14,20,1,12,5,3,1,2,4,9,1,10,13,15,1,2,3,6,1,7,5,4,1,2,8,11,1,3,17,$\\
$19,1,2,21$
\end{footnotesize}

\medskip\noindent
\textbf{A periodic packing 16-coloring of $D(1,4)$ of period $320$}\\
\begin{footnotesize}
$1,2,1,3,4,1,5,1,2,7,1,6,1,3,2,1,8,1,4,10,1,2,1,3,5,1,9,1,2,12,1,13,1,3,2,1,4,1,6,7,1,2,1,3,11,1,5,1,2,$
$8,1,4,1,3,2,1,14,1,10,15,1,2,1,3,5,1,4,1,2,6,1,7,1,3,2,1,9,1,12,8,1,2,1,3,4,1,5,1,2,11,1,6,1,3,2,1,10,$
$1,4,13,1,2,1,3,5,1,7,1,2,8,1,9,1,3,2,1,4,1,6,14,1,2,1,3,12,1,5,1,2,15,1,4,1,3,2,1,7,1,10,8,1,2,1,3,5,1,$
$4,1,2,6,1,9,1,3,2,1,11,1,13,16,1,2,1,3,4,1,5,1,2,7,1,6,1,3,2,1,8,1,4,10,1,2,1,3,5,1,9,1,2,12,1,14,1,3,$
$2,1,4,1,6,7,1,2,1,3,11,1,5,1,2,8,1,4,1,3,2,1,13,1,10,15,1,2,1,3,5,1,4,1,2,6,1,7,1,3,2,1,9,1,12,8,1,2,1,$
$3,4,1,5,1,2,11,1,6,1,3,2,1,10,1,4,14,1,2,1,3,5,1,7,1,2,8,1,9,1,3,2,1,4,1,6,13,1,2,1,3,12,1,5,1,2,15,1,4,$
$1,3,2,1,7,1,10,8,1,2,1,3,5,1,4,1,2,6,1,9,1,3,2,1,11,1,14,16$
\end{footnotesize}

\medskip\noindent
\textbf{A periodic packing 13-coloring of $D(2,3)$ of period $240$}\\
\begin{footnotesize}
$1,1,2,3,4,1,1,5,6,2,1,1,8,3,13,1,1,2,4,11,1,1,7,3,2,1,1,5,6,9,1,1,2,3,4,1,1,8,10,2,1,1,12,3,5,1,1,2,$\\
$4,6,1,1,7,3,2,1,1,9,11,13,1,1,2,3,4,1,1,5,6,2,1,1,8,3,7,1,1,2,4,10,1,1,12,3,2,1,1,5,6,9,1,1,2,3,4,1,$\\
$1,7,8,2,1,1,11,3,5,1,1,2,4,6,1,1,10,3,2,1,1,9,13,7,1,1,2,3,4,1,1,5,6,2,1,1,8,3,12,1,1,2,4,11,1,1,7,3,$\\
$2,1,1,5,6,9,1,1,2,3,4,1,1,8,10,2,1,1,13,3,5,1,1,2,4,6,1,1,7,3,2,1,1,9,11,12,1,1,2,3,4,1,1,5,6,2,1,1,$\\
$8,3,7,1,1,2,4,10,1,1,13,3,2,1,1,5,6,9,1,1,2,3,4,1,1,7,8,2,1,1,11,3,5,1,1,2,4,6,1,1,10,3,2,1,1,9,12,7$
\end{footnotesize}

\medskip\noindent
\textbf{A periodic packing 23-coloring of $D(2,5)$ of period $336$}\\
\begin{footnotesize}
$1,1,2,2,1,3,4,1,1,5,6,1,7,8,1,1,2,2,1,3,10,1,1,11,4,1,15,12,1,1,2,2,1,3,16,1,1,5,6,1,4,9,1,1,2,2,1,3,$\\
$7,1,1,8,14,1,17,13,1,1,2,2,1,3,4,1,1,5,6,1,10,19,1,1,2,2,1,3,11,1,1,7,4,1,9,12,1,1,2,2,1,3,8,1,1,5,6,$\\
$1,4,15,1,1,2,2,1,3,18,1,1,20,21,1,7,22,1,1,2,2,1,3,4,1,1,5,6,1,10,9,1,1,2,2,1,3,8,1,1,11,4,1,13,12,1,1,$\\
$2,2,1,3,7,1,1,5,6,1,4,14,1,1,2,2,1,3,16,1,1,17,23,1,9,19,1,1,2,2,1,3,4,1,1,5,6,1,7,8,1,1,2,2,1,3,10,1$\\
$,1,11,4,1,15,12,1,1,2,2,1,3,13,1,1,5,6,1,4,9,1,1,2,2,1,3,7,1,1,8,18,1,14,20,1,1,2,2,1,3,4,1,1,5,6,1,10,$\\
$21,1,1,2,2,1,3,11,1,1,7,4,1,9,12,1,1,2,2,1,3,8,1,1,5,6,1,4,13,1,1,2,2,1,3,15,1,1,16,17,1,7,19,1,1,2,2,$\\
$1,3,4,1,1,5,6,1,10,9,1,1,2,2,1,3,8,1,1,11,4,1,14,12,1,1,2,2,1,3,7,1,1,5,6,1,4,18,1,1,2,2,1,3,13,1,1,20,$\\
$22,1,9,23$
\end{footnotesize}

\section{An illustration of Subcase A.3 of the proof of Proposition~\ref{p0}}
\label{a2}
We illustrate the construction of a packing coloring of $D(1,t)$ defined in Subcase A.3 for $t=75$ or $t=77$, i.e. $s=76 = 48+28$.

The color patterns $R^i_j$ are defined as follows:\\
$R_1^i = (1,4,1,5,1,8,1,4,1,5,1,9)^6 . (1,32+i,1,35+i)$, $0\leq i\le 2$;\\
$R_2^i = (1,4,1,5,1,10,1,4,1,5,1,11)^6 . (1,38+i,1,41+i)$, $0\leq i\le 2$;\\
$R_3^i = (1,6,1,7,1,12,1,13,1,6,1,7,1,14,1,15)^4 . (1,6,1,7,1,44+i,1,48+i,1,52+i,1,56+i)$,\\
\hspace*{.95cm}$0\leq i\le 3$;\\
$R_4^i = (1,6,1,7,1,16,1,17,1,6,1,7,1,18,1,19,1,6,1,7,1,20,1,21)^3 .$\\
\hspace*{.95cm}$(1,60+i,1,62+i), 0\leq i\le 1$; \\
$R_5^i = (1,6,1,7,1,22,1,23,1,6,1,7,1,24,1,25,1,6,1,7,1,26,1,27,1,6,1,7,1,28,1,29) .$\\
\hspace*{.95cm}$(1,6,1,7,1,64+i,1,67+i,1,70+i,1,73+i)$, $0\leq i\le 2$;

And a packing $75$-coloring is obtained by assigning to subgraphs $B_{2i+1}$ the color pattern $(1,2,1,3)^{19}$ and repeatedly to $48$ consecutive subgraphs $B_{2i}$ the sequence of color patterns
\begin{align*}
\mathcal{R} = (&
R_{1}^0, R_{3}^0, R_{2}^0, R_{4}^0, R_{1}^1, R_{3}^1, R_{2}^1, R_{5}^0,
R_{1}^2, R_{3}^2, R_{2}^2, R_{4}^1, R_{1}^0, R_{3}^3, R_{2}^0, R_{5}^1,
R_{1}^1, R_{3}^0, R_{2}^1, R_{4}^0, R_{1}^2, R_{3}^1, R_{2}^2, R_{5}^2,\\
& 
R_{1}^0, R_{3}^2, R_{2}^0, R_{4}^1, R_{1}^1, R_{3}^3, R_{2}^1, R_{5}^0,
R_{1}^2, R_{3}^0, R_{2}^2, R_{4}^0, R_{1}^0, R_{3}^1, R_{2}^0, R_{5}^1,
R_{1}^1, R_{3}^2, R_{2}^1, R_{4}^1, R_{1}^2, R_{3}^3, R_{2}^2, R_{5}^2).
\end{align*}

\end{document}